\documentclass[11 pt]{article}
\usepackage{amsmath, mathtools}
\usepackage{amsthm}
\usepackage{amssymb}
\usepackage{graphicx}
\usepackage{appendix}
\usepackage{color}
\usepackage{float}
\usepackage{hyperref}
\usepackage{bbm}
\usepackage{easyReview}

\newtheorem{definition}{Definition}[section]

\newtheorem{theorem}[definition]{Theorem}

\newtheorem{corollary}[definition]{Corollary}
\newtheorem{remark}[definition]{Remark}
\numberwithin{equation}{section}

\begin{document}
\title{Estimation of VaR with jump process: application in corn and soybean markets}

\author{ Minglian Lin\footnote{Department of Mathematics, North Dakota State University, Email: minglian.lin@ndsu.edu}, \quad Indranil SenGupta\footnote{Environmental Finance \& Risk Management in the Institute of Environment,  and Department of Mathematics and Statistics, Florida International University, Email: isengupt@fiu.edu}, \quad William Wilson\footnote{CHS Chair in Risk Management and Trading; University Distinguished Professor, Department of Agribusiness and Applied Economics, North Dakota State University, Email: william.wilson@ndsu.edu}}

\date{\today}

\maketitle
\begin{abstract}

\noindent Value at Risk (VaR) is a quantitative measure used to evaluate the risk linked to the potential loss of investment or capital.  Estimation of the VaR entails the quantification of prospective losses in a portfolio of investments, using a certain likelihood, under normal market conditions within a specific time period. The objective of this paper is to construct a model and estimate the VaR for a diversified portfolio consisting of multiple cash commodity positions driven by standard Brownian motions and jump processes. Subsequently, a thorough analytical estimation of the VaR is conducted for the proposed model. The results are then applied to two distinct commodities- corn and soybean- enabling a comprehensive comparison of the VaR values in the presence and absence of jumps.

\end{abstract}



\section{Introduction}
	
Value at Risk (VaR) is a quantitative metric utilized to assess the potential risk associated with the loss of investment or capital. Estimation of VaR quantifies the potential loss of a portfolio of investments, based on a specified probability, under normal market conditions within a certain time frame, such as a single day. VaR is commonly employed by trading firms, financial institutions and regulatory agencies as a means of assessing the requisite amount of assets necessary to mitigate potential losses; and by trading companies as a tool for managing risk. 

The usefulness of VaR arises from the significant fluctuations observed in exchange rates, interest rates, and commodity prices during recent decades, as well as the widespread utilization of derivative instruments for mitigating risks associated with price fluctuations. With the spread of derivatives came increased trading of cash instruments and securities, as well as an expansion of financing options. Consequently, numerous corporations have portfolios comprised of substantial quantities of cash and derivative instruments, which may occasionally exhibit complicated characteristics. Furthermore, the extent of the risks included in the portfolios of organizations is frequently not readily apparent. There is a growing need for a quantitative assessment of market risk at the portfolio level. Many methods exist for calculating VaR, including historical simulation, the delta-normal approach, and Monte Carlo simulation, among others (see \cite{Thomas}). 

The paper \cite{Peter} presents a VaR analysis framework. The analysis in \cite{Peter} deals with, first, to verify the accuracy of a specified VaR measure; and second, to statistically distinguish between two models in order to determine the superior one. Comparative VaRs are computed within an application using historical or option-price-based volatility measures. Risk measurement is, in essence, the process of determining the distribution of risks. In that regard, Conditional value-at-risk (CVaR) and VaR are well-known functions utilized in risk management. In the literature, the dilemma of selecting between VaR and CVaR has been quite prevalent and is influenced by various factors, including disparities in mathematical properties, the reliability of statistical estimation, the ease of optimizing procedures, regulatory acceptability, and more. The paper \cite{Stan} describes the merits and demerits of these risk measures.  It is worth noting that \cite{Stan2} introduces and evaluates a novel methodology for enhancing the performance and mitigating the risk of a portfolio. The primary emphasis is placed on the minimization of CVaR rather than VaR. However, it should be noted that portfolios exhibiting low CVaR inherently exhibit low VaR as well. Other empirical measures commonly used in commodity risk measurement and management include  a simple standard deviation, coefficient of variation, semi variance, in addition to the omega ratio, among others (see \cite{Acerbi1, Artzner1, Follmer1}).

Hybrid methodologies that include AR-GARCH filtering with skewed-t residuals and the extreme value theory-based approach are highly recommended among the repertoire of recognized benchmark predictors for VaR. In recent years, the paper \cite{Peng} presents a new VaR predictor called G-VaR. The G-VaR predictor has been extensively tested on both the NASDAQ Composite Index and S\&P500 Index, revealing its good performance compared to many benchmark VaR predictors.  Similarly, 
in \cite{stoch3} the authors introduce a novel distribution for quantifying the size-of-loss in insurance claims data that exhibits a negative skew. The autoregressive model is suggested as a means of examining insurance claims data and predicting future values of predicted claims. The utilization of VaR estimate is employed as a primary instrument in this investigation.

The paper \cite{stoch1} provides an application of VaR where a stochastic model is implemented to analyze the participation of virtual power plants (VPPs) in futures markets, pool markets, and contracts with withdrawal penalties.  Another application is provided in \cite{stoch2}, where a generalized autoregressive conditional heteroscedasticity (GARCH) based VaR modeling is implemented for some precious metals.  In order to evaluate the fluctuating dynamics in the estimation of VaR, the paper \cite{stoch5}  utilize an integrated methodology that combines dynamic conditional correlation (DCC) and GARCH models. The analysis focuses on the daily stock returns of emerging markets. On the other hand, a corresponding portfolio optimization problem is considered in \cite{LinS, LinS2}. The study examines the portfolio optimization problem in a financial market, specifically focusing on a broad utility function.

There are some recent papers where corn and soybean markets are analyzed with stochastic model. For instance, the paper \cite{Shan} presents a comprehensive model for analyzing the dynamics of market share in the soybean export industry. Additionally, theoretical calculations specifically focused on a particular scenario inside the broader model is conducted. On the other hand,  there exists a significant level of competition between the United States and Brazil in the Chinese soybean import market. The volatility and risk associated with logistical functions and costs have a significant impact on the export competition between the two countries.  The paper \cite{Gwen} examines the methods employed in commodities trading and investigates the impact of logistical functions and costs on shipments from the United States and Brazil to China. The analysis is conducted using an Optimized Monte Carlo Simulation model, which takes into consideration a significant number of random and correlated factors.  The paper \cite{Humyr} presents a comprehensive mathematical model for the analysis of yield data, that are derived from a representative corn field located in the upper midwestern region of the United States. Expressions for statistical moments are derived from the underlying stochastic model.

An analytical estimation of VaR for multi-asset portfolio is solved from the multi-asset stock pricing formula driven by the conventional geometric Brownian motion \cite{FSU}. In this paper, we extend this multivariate VaR model by considering the jumps of multi-asset stock price. As a result, the purpose of the present paper is to develop (specify) a model of VaR for a multi-asset portfolio with multiple cash positions driven by standard Brownian motions and jump processes. After that we analytically estimate the VaR for the proposed model, apply the results to two commodities, and compare the VaR with and without jumps. 

The rest of the paper is structured as follows: In Section \ref{sec2} we improve multi-asset model of asset price using jump process and derive the general formula of the expectation of VaR with jumps for a multi-asset portfolio. In Section \ref{sec3} we provide an application of the proposed model  of VaR to the daily data of cash value of corn and soybean. A brief conclusion is provided in Section \ref{sec4}.

\section{Multi-asset model of asset price with Brownian motion and jump process}
\label{sec2}
	We assume that a portfolio contains $ n $ correlated cash assets.  Let $ S_i(t)$, $ i = 1\dots n $, be the price of assets at time $ t \in [0,T] $ where $ T $ is terminal time. 
	We consider that there are significant unexpected discontinuous changes in asset prices.
	For stochastic modeling, we assume that there are $ m $ independent jump processes in the price of each asset.
	With the notations in \cite{Oksendal_1}, we denote the compensated jump measures by $	\tilde{N}_k(dt,d\zeta_k) \equiv N_k(dt,d\zeta_k)  -  \nu_k(d\zeta_k)dt$, for $k = 1,\dots,m$, where $ N_k(dt,d\zeta_k) $ is the differential jump measure (Poisson random measure) giving the number of jumps through $ dt $ with differential generic jump size $ d\zeta_k \subset \mathbb{R} \setminus \{0\} $, and $ \nu(\cdot) $ is the L\'evy measure defined by $ \nu_k(d\zeta_k) = \mathbb{E}\big(N(1,d\zeta_k)\big) $ where $ \mathbb{E} $ is expectation function.
	In addition, we denote $ \boldsymbol{\zeta} = (\zeta_1, \dots, \zeta_m)^\top $.

\subsection{Asset price dynamics}	
	We note the geometric Brownian motion implies that the asset prices are non-negative, and the percentage changes in asset price are random and independent for all increments in time. 
	Hence, we define the price of assets by the following generalized geometric processes
	\begin{align}\label{geo}
		\frac{dS_i(t)}{S_i(t)} = 
		\mu_i dt +  \sum_{j=1}^n a_{ij} dW_j(t) + \sum_{k=1}^m \int_\mathbb{R} \gamma_{ik}(t,\boldsymbol{\zeta})\tilde{N}_k(dt,d\zeta_k),
		\quad S_i(0) >0, 
	\end{align}
	where $i = 1, \dots, n$,
	$ \frac{dS_i(t)}{S_i(t)} $ is instantaneous return,
	$ \mu_i $ is constant expectation of return after removing all the jumps from asset prices, 
	$ W_j(t) $ is standard Brownian motion,
	and $ \gamma_{ik}(t,\boldsymbol{\zeta}) $ is predictable process.
	Meanwhile, $ (a_{ij})_{1\leq i,j \leq n} = \textbf{A} $ is the lower triangular Cholesky decomposition of the constant covariance matrix of return after removing all the jumps from asset prices, i.e. $  \boldsymbol{\Sigma} = \textbf{A}\textbf{A}^\top  $, for computational efficiency  (see \cite{FSU, Glasserman}).  As the estimation of VaR is considered for the short-term future, we assume constant volatility $ a_{ij} $, $ 1\leq i,j \leq n $.
	
	If $ \gamma_{ik} > -1 $, then $ S_i(t) $ can never jump to $ 0 $ or a negative value (see Example 2.5, \cite{Oksendal_1}). 
	To find the solution of stochastic differential equation \eqref{geo}, for $ i = 1,\dots, n $, we firstly apply It\^o formula of  jump process (Theorem 1.14 in \cite{Oksendal_2}) to $ \ln S_i(t) $ and obtain
	\begin{align}
		& d \ln S_i(t)  
		\nonumber \\
		&= \frac{S_i(t)}{S_i(t)} \bigg[ \mu_i dt + \sum_{j=1}^n a_{ij} dW_j(t) \bigg]
		- \frac{1}{2} \frac{S_i(t)^2}{S_i(t)^2} \sum_{j=1}^n a_{ij}^2 dt \nonumber \\
		& \quad 
		+ \sum_{k=1}^m \int_\mathbb{R} \bigg[ \ln \big(S_i(t) + \gamma_{ik}(t,\boldsymbol{\zeta})S_i(t)\big) - \ln S_i(t) -  \frac{S_i(t)}{S_i(t)} \gamma_{ik}(t,\boldsymbol{\zeta}) \bigg] \nu_k(d\zeta_k) dt \nonumber \\
		& \quad
		+ \sum_{k=1}^m \int_\mathbb{R} \bigg[ \ln \big(S_i(t) + \gamma_{ik}(t,\boldsymbol{\zeta})S_i(t)\big) - \ln S_i(t) \bigg] \tilde{N}_k(dt,d\zeta_k) \nonumber \\
		&=  \Big[\mu_i-\frac{\sigma_i^2}{2}\Big]dt + \sum_{j=1}^n a_{ij} W_j(t)
		+ \sum_{k=1}^m \int_\mathbb{R} \Big[ \ln \big(1 + \gamma_{ik}(t,\boldsymbol{\zeta})\big) - \gamma_{ik}(t,\boldsymbol{\zeta}) \Big] \nu_k(d\zeta_k) dt \label{Cholesky} \nonumber \\ 
		& \quad
		+ \sum_{k=1}^m \int_\mathbb{R} \ln \big(1 + \gamma_{ik}(t,\boldsymbol{\zeta})\big)  \tilde{N}_k(dt,d\zeta_k) \\
		&=  \Big[\mu_i-\frac{\sigma_i^2}{2}\Big]dt + \sum_{j=1}^n a_{ij} W_j(t)
		+ \sum_{k=1}^m \int_\mathbb{R} - \gamma_{ik}(t,\boldsymbol{\zeta}) \nu_k(d\zeta_k) dt  + \sum_{k=1}^m \int_\mathbb{R} \ln \big(1 + \gamma_{ik}(t,\boldsymbol{\zeta})\big) N_k(dt,d\zeta_k). \nonumber
	\end{align}
	Here, equation \eqref{Cholesky} makes use of the fact $ \sum_{j=1}^n a_{ij}^2 = \sigma_i^2 $ by definition of Cholesky decomposition, where the scalar $ \sigma_i $ is constant standard deviation of return after removing all the jumps from the $i$-th asset price.
	After that, taking integral to both sides of above equation with $ W_j(0) = 0 $, for all $j$, gives the following geometric L\'evy process
	\begin{align}\label{geoL}
		S_i(t)
		= S_i(0) \exp \bigg(
		& \Big[\mu_i-\frac{\sigma_i^2}{2}\Big]t + \sum_{j=1}^n a_{ij} W_j(t)
		+ \sum_{k=1}^m \int_0^t \int_\mathbb{R} - \gamma_{ik}(s,\boldsymbol{\zeta}) \nu_k(d\zeta_k) ds \nonumber \\
		& + \sum_{k=1}^m \int_0^t \int_\mathbb{R} \ln \big(1 + \gamma_{ik}(s,\boldsymbol{\zeta})\big) N_k(ds,d\zeta_k)
		\bigg). 
	\end{align}

	\subsection{Estimation of VaR with jump process}
	
	For a random variable 
	\begin{equation}
	\label{dn}
	 Y = \exp(\tilde{\mu}+\tilde{\sigma} Z),
	 \end{equation}
	 where $ Z \sim \mathcal{N}(0, 1) $, the Value-at-Risk (VaR) is defined by 
	 \begin{equation}	 
	 \label{dn2}
	  \text{VaR}_\alpha(Y) = \sup \{x: \mathbb{P}(Y \leq x) = 1-\alpha\},
	  \end{equation}
	   where $ \mathbb{P} $ is probability measure, at confidence level $ \alpha $. A typical value of $\alpha$ is close to (but not greater than) $ 1 $. It is usually $ 0.95 $. Therefore, the significance level is given in the form $ 1-\alpha $ in this study.
	It follows that 
	\begin{align*}
		\mathbb{P}(Y\leq x) = \mathbb{P}\big(\exp(\tilde{\mu}+\tilde{\sigma} Z)\leq x\big) = \mathbb{P}\Big(Z \leq\frac{\ln x - \tilde{\mu}}{\tilde{\sigma}}\Big)
		:= \Phi\Big(\frac{\ln x - \tilde{\mu}}{\tilde{\sigma}}\Big) = 1-\alpha,
	\end{align*}
	where $\Phi(\cdot)$ is the cdf for the standard normal distribution. For \emph{fixed} $ t $, we define $ W_j(t) = \sqrt{t}Z_j $, $ j = 1, \dots, n $, where $ \{Z_j\}_{j=1}^n \sim \mathcal{N}(0, 1) $. 
	For a fixed $t$, by substituting it in the formula \eqref{geoL}, we note
	\begin{align*}
		\sum_{j=1}^n a_{ij} W_j(t) = \sqrt{t} \sum_{j=1}^n a_{ij} Z_j = \sqrt{t} \boldsymbol{a}_i \boldsymbol{Z},
	\end{align*}
	where $ \boldsymbol{a}_i = (a_{i1}, \dots, a_{in}) $ and $ \boldsymbol{Z} = (Z_1, \dots , Z_n)^\top $.
	We consider that a jump of asset price always changes the mean of asset price significantly, for the follow-up days.
	We will use the random variable $Y$ in \eqref{dn} and the corresponding VaR in \eqref{dn2} as motivation. Comparing  \eqref{geoL} with \eqref{dn} we obtain for $ i = 1\dots n $, 
	
	\begin{align*}
		 \tilde{\mu}_i {} = \
		& \Big[\mu_i-\frac{\sigma_i^2}{2}\Big]t 
		+ \sum_{k=1}^m \int_0^t \int_\mathbb{R} - \gamma_{ik}(s,\boldsymbol{\zeta}) \nu_k(d\zeta_k) ds
		+ \sum_{k=1}^m \int_0^t \int_\mathbb{R} \ln \big(1 + \gamma_{ik}(s,\boldsymbol{\zeta})\big) N_k(ds,d\zeta_k), \\
		\tilde{\sigma}_i{} = \
		& \sqrt{t} \boldsymbol{a}_i.
	\end{align*}
	In the case of $ n $ assets, let $ \boldsymbol{w} = (w_1, \dots, w_n)^\top $ be the vector of weight of each asset such that $ \sum_{i=1}^{n} w_i =1 $.  Including the weights for multiple assets, we define the scalars $ \tilde{\mu} $ and $ \tilde{\sigma} $ with jump process as:
	\begin{align*}
		\tilde{\mu} {} = \
		& \boldsymbol{w}^\top \bigg[ \Big[\boldsymbol{\mu} - \frac{1}{2} \boldsymbol{\sigma}\circ\boldsymbol{\sigma}\Big]t 
		+ \sum_{k=1}^m \int_0^t \int_\mathbb{R} - \boldsymbol{\gamma}_{k}(s,\boldsymbol{\zeta}) \nu_k(d\zeta_k) ds
		+ \sum_{k=1}^m \int_0^t \int_\mathbb{R} \boldsymbol{f}_k(s,\boldsymbol{\zeta}) N_k(ds,d\zeta_k)
		\bigg], \\
		\tilde{\sigma} {} = \
		& \big( \boldsymbol{w}^\top \big(\sqrt{t} \textbf{A} ( \sqrt{t} \textbf{A})^\top \big) \boldsymbol{w} \big)^{\frac{1}{2}} 
		= \sqrt{t}\big( \boldsymbol{w}^\top \boldsymbol{\Sigma} \boldsymbol{w} \big)^{\frac{1}{2}} .
	\end{align*}
	where $ \boldsymbol{\mu} = (\mu_1, \dots, \mu_n)^\top $, $ \boldsymbol{\sigma} = (\sigma_1, \dots, \sigma_n)^\top $, 
	$ \boldsymbol{f}_k(s,\boldsymbol{\zeta}) = \big(\ln(1 + \gamma_{1k}(s,\boldsymbol{\zeta})), \dots, \ln(1 + \gamma_{nk}(s,\boldsymbol{\zeta}))\big)^\top $,
	$ \boldsymbol{\gamma}_{k}(s,\boldsymbol{\zeta}) = \big( \gamma_{1k}(s,\boldsymbol{\zeta}),\dots, \gamma_{nk}(s,\boldsymbol{\zeta}) \big)^\top $,
	$ \circ $ is Hadamard product, and
	$ \textbf{A} = (\boldsymbol{a}_1, \dots, \boldsymbol{a}_n)^\top $ is the lower triangular Cholesky decomposition of the constant covariance matrix of return, i.e. $  \boldsymbol{\Sigma} = \textbf{A}\textbf{A}^\top  $.
	
We define the VaR of a multi-asset portfolio with multiple cash positions as:
	\begin{align*}
		\text{VaR}_\alpha(t) 
		{}= \
		& \exp \Bigg(
		\boldsymbol{w}^\top \bigg[ \Big[\boldsymbol{\mu} - \frac{1}{2} \boldsymbol{\sigma}\circ\boldsymbol{\sigma}\Big]t 
		+ \sum_{k=1}^m \int_0^t \int_\mathbb{R} - \boldsymbol{\gamma}_{k}(s,\boldsymbol{\zeta}) \nu_k(d\zeta_k) ds \nonumber \\
		& \quad \quad \ \ 
		+ \sum_{k=1}^m \int_0^t \int_\mathbb{R} \boldsymbol{f}_k(s,\boldsymbol{\zeta}) N_k(ds,d\zeta_k)
		\bigg]
		+ \sqrt{t}\big( \boldsymbol{w}^\top \boldsymbol{\Sigma} \boldsymbol{w} \big)^{\frac{1}{2}} \Phi^{-1}(1-\alpha)
		\Bigg).
	\end{align*}
It follows that, we can define:
\begin{definition}
\label{def1}
	\begin{align*}
		\text{VaR}_\alpha(t) 
		= \
		& \exp \bigg( \boldsymbol{w}^\top \Big[\boldsymbol{\mu} - \frac{1}{2} \boldsymbol{\sigma}\circ\boldsymbol{\sigma}\Big]t
		\bigg) 
		\prod_{k=1}^m \exp \bigg( \int_0^t \int_\mathbb{R} - \boldsymbol{w}^\top \boldsymbol{\gamma}_{k}(s,\boldsymbol{\zeta}) \nu_k(d\zeta_k) ds \bigg) \\
		&
		\prod_{k=1}^m \exp \bigg( \int_0^t \int_\mathbb{R} \boldsymbol{w}^\top \boldsymbol{f}_k(s,\boldsymbol{\zeta}) N_k(dt,d\zeta_k)
		\bigg)
		\exp \bigg(\sqrt{t}\big( \boldsymbol{w}^\top \boldsymbol{\Sigma} \boldsymbol{w} \big)^{\frac{1}{2}} \Phi^{-1}(1-\alpha)
		\bigg).
	\end{align*}
\end{definition}

	In Definition \ref{def1}, note that $ \exp \big( \int_0^t \int_\mathbb{R} \boldsymbol{w}^\top \boldsymbol{f}_k(s,\boldsymbol{\zeta}) N_k(ds,d\zeta_k) \big) $, $ k = 1, \dots, m $, are independent stochastic terms while the remaining terms are deterministic.
	This implies that, for \emph{fixed} $ t $, $ \text{VaR}_\alpha(t) $ is stochastic with respect to Poisson random measure $ N $.
	However, for empirical application we would like to fins a concrete value. Consequently, we find  the expected value of $ \text{VaR}_\alpha $ as: 
	\begin{align}
	\label{dn3}
		\mathbb{E} (\text{VaR}_\alpha) (t)
		= \
		& \exp \bigg( \boldsymbol{w}^\top \Big[\boldsymbol{\mu} - \frac{1}{2} \boldsymbol{\sigma}\circ\boldsymbol{\sigma}\Big]t
		\bigg)
		\prod_{k=1}^m \exp \bigg( \int_0^t \int_\mathbb{R} - \boldsymbol{w}^\top \boldsymbol{\gamma}_{k}(s,\boldsymbol{\zeta}) \nu_k(d\zeta_k) ds
		\bigg) \nonumber \\
		&
		\prod_{k=1}^m \mathbb{E} \bigg( \exp \bigg( \int_0^t \int_\mathbb{R} \boldsymbol{w}^\top \boldsymbol{f}_k(s,\boldsymbol{\zeta})  N_k(ds,d\zeta_k)
		\bigg) \bigg)
		\exp \bigg( \sqrt{t}\big( \boldsymbol{w}^\top \boldsymbol{\Sigma} \boldsymbol{w} \big)^{\frac{1}{2}} \Phi^{-1}(1-\alpha)
		\bigg).
	\end{align}

\begin{theorem} \label{THM}
		In space $ \big([0, \infty)\times (\mathbb{R}\setminus\{0\}), \mathcal{A} \times \mathcal{B}, N(t,\zeta)\big) $ where $ N $ is jump measure,
		$ \forall f \in L^1 $,
		\begin{align*}
			\mathbb{E} \bigg( \exp \bigg( \int_0^t \int_\mathbb{R} f(s,\zeta)  N(ds,d\zeta) \bigg) \bigg)
			= \exp\bigg( \int_0^t \int_\mathbb{R} \Big[\exp\big(f(s,\zeta)\big)-1\Big] \nu(d\zeta)ds \bigg).
		\end{align*}
	\end{theorem}

The proof is provided in Appendix \ref{app1koy}. Implementing \eqref{dn3} and Theorem \ref{THM}, we obtain the estimate of VaR as the following theorem:
	
\begin{theorem}
When $\text{VaR}_\alpha(t)$ is defined as in the Definition \ref{def1}, 
	\begin{align} \label{VaR}
		\mathbb{E} (\text{VaR}_\alpha) (t)= & \exp \bigg( 
		\boldsymbol{w}^\top \Big[\boldsymbol{\mu} - \frac{1}{2} \boldsymbol{\sigma}\circ\boldsymbol{\sigma}\Big]t
		+ \sqrt{t}\big( \boldsymbol{w}^\top \boldsymbol{\Sigma} \boldsymbol{w} \big)^{\frac{1}{2}} \Phi^{-1}(1-\alpha)
		\bigg) \nonumber \\
		&
		\prod_{k=1}^m \exp\bigg( \int_0^t \int_\mathbb{R} \Big[\exp\big(\boldsymbol{w}^\top \boldsymbol{f}_k(s,\boldsymbol{\zeta})\big) - 1
		- \boldsymbol{w}^\top \boldsymbol{\gamma}_{k}(s,\boldsymbol{\zeta})\Big] \nu_k(d\zeta_k)ds \bigg).
	\end{align}
\end{theorem}

We now observe the following interesting result. 
	\begin{theorem} \label{Thm}
		The jumps of asset price do not explicitly exist in the expectation of value-at-risk, in either of the following two cases:
		\begin{enumerate}
			\item the portfolio contains only one asset;
			\item for the $ k $-th jump, $ k = 1, \dots, m $, the returns of all the $ n $ assets in a portfolio follow the same predictable process $ \gamma_{k} $, i.e. $ \gamma_{k} = \gamma_{1k} = \gamma_{2k} = \dots = \gamma_{nk} $.
		\end{enumerate}
	\end{theorem}
	\begin{proof}
		In the single asset case, we start from formula \eqref{VaR} which is the general formula of $ \mathbb{E} \big(\text{VaR}_\alpha(t)\big) $. By subsituting $ f(s,\zeta) = \ln(1 + \gamma(s,\zeta)) $, we have: $ \exp\big(f(s,\zeta)\big) - 1 - \gamma(s,\zeta) = 0 $. Hence
		\begin{align*}
			\mathbb{E} (\text{VaR}_\alpha)(t)
			{} = \ 
			& \exp \bigg( \Big[\mu - \frac{\sigma^2}{2} \Big]t + \sigma \Phi^{-1}(1-\alpha) \sqrt{t} \bigg)
			\exp\bigg( \int_0^t \int_\mathbb{R} 0\ \nu(d\zeta)ds \bigg) \nonumber \\
			{} = \
			& \exp \bigg( \Big[\mu - \frac{\sigma^2}{2} \Big]t + \sigma \Phi^{-1}(1-\alpha) \sqrt{t} \bigg).
		\end{align*}
		
		In the $ 2 $nd case, we start from the integrand of the jump terms of formula \eqref{VaR}.
		For every $ k = 1, \dots m $, subsituting $ \boldsymbol{f}_k(s,\boldsymbol{\zeta}) = \big(\ln(1 + \gamma_{1k}(s,\boldsymbol{\zeta})), \dots, \ln(1 + \gamma_{nk}(s,\boldsymbol{\zeta}))\big)^\top $ gives
		\begin{align*}
			\exp\big(\boldsymbol{w}^\top \boldsymbol{f}_k(s,\boldsymbol{\zeta})\big) - 1
			- \boldsymbol{w}^\top \boldsymbol{\gamma}_{k}(s,\boldsymbol{\zeta})
			& = \exp\bigg( \sum_{i=1}^n \ln\big(1 + \gamma_{ik}(s,\boldsymbol{\zeta})\big)^{w_i} \bigg) - 1
			- \sum_{i=1}^n w_i \gamma_{ik}(s,\boldsymbol{\zeta}) \\
			& = \prod_{i=1}^n \big(1 + \gamma_{ik}(s,\boldsymbol{\zeta})\big)^{w_i} - 1
			- \sum_{i=1}^n w_i \gamma_{ik}(s,\boldsymbol{\zeta}).
		\end{align*}
		If $ \gamma_{k}(s,\boldsymbol{\zeta}) = \gamma_{1k}(s,\boldsymbol{\zeta}) = \gamma_{2k}(s,\boldsymbol{\zeta}) = \dots = \gamma_{nk}(s,\boldsymbol{\zeta}) $,
		then by $ \sum_{i=1}^n w_i = 1 $ defined before,
		\begin{align*}
			\exp\big(\boldsymbol{w}^\top \boldsymbol{f}_k(s,\boldsymbol{\zeta})\big) - 1
			- \boldsymbol{w}^\top \boldsymbol{\gamma}_{k}(s,\boldsymbol{\zeta})
			= \Big(1 + \gamma_{k}(s,\boldsymbol{\zeta})\Big)^{\sum_{i=1}^n w_i} - 1
			- \gamma_{k}(s,\boldsymbol{\zeta}) \sum_{i=1}^n w_i
			= 0.
		\end{align*}
		Hence
		\begin{align*} 
			\mathbb{E} (\text{VaR}_\alpha) (t)
			{} = 
			& \exp \bigg( 
			\boldsymbol{w}^\top \Big[\boldsymbol{\mu} - \frac{1}{2} \boldsymbol{\sigma}\circ\boldsymbol{\sigma}\Big]t
			+ \sqrt{t}\big( \boldsymbol{w}^\top \boldsymbol{\Sigma} \boldsymbol{w} \big)^{\frac{1}{2}} \Phi^{-1}(1-\alpha)
			\bigg)
			\prod_{k=1}^m \exp\bigg( \int_0^t \int_\mathbb{R} 0\ \nu_k(d\zeta_k)ds \bigg) \nonumber \\
			{} = 
			& \exp \bigg( 
			\boldsymbol{w}^\top \Big[\boldsymbol{\mu} - \frac{1}{2} \boldsymbol{\sigma}\circ\boldsymbol{\sigma}\Big]t
			+ \sqrt{t}\big( \boldsymbol{w}^\top \boldsymbol{\Sigma} \boldsymbol{w} \big)^{\frac{1}{2}} \Phi^{-1}(1-\alpha)
			\bigg).
		\end{align*}
	\end{proof}

	\begin{corollary}
		Under Case 1 of Theorem \ref{Thm},
		\begin{align} \label{VaR-1}
			\mathbb{E} (\text{VaR}_\alpha)(t)
			=
			\exp \bigg( \Big[\mu - \frac{\sigma^2}{2} \Big]t + \sigma \Phi^{-1}(1-\alpha) \sqrt{t} \bigg).
		\end{align}
		Under Case 2 of Theorem \ref{Thm},
		\begin{align} \label{VaR-2}
			\mathbb{E} (\text{VaR}_\alpha) (t)
			= 
			\exp \bigg( 
			\boldsymbol{w}^\top \Big[\boldsymbol{\mu} - \frac{1}{2} \boldsymbol{\sigma}\circ\boldsymbol{\sigma}\Big]t
			+ \sqrt{t}\big( \boldsymbol{w}^\top \boldsymbol{\Sigma} \boldsymbol{w} \big)^{\frac{1}{2}} \Phi^{-1}(1-\alpha)
			\bigg).
		\end{align}
	\end{corollary}

	\begin{remark}

		Although there is not jump term in formulas \eqref{VaR-1} and \eqref{VaR-2}, the sample mean vector of return $ \boldsymbol{\mu} $, sample standard deviation vector of return $ \boldsymbol{\sigma} $, and sample covariance matrix of return $ \boldsymbol{\Sigma} $ are calculated after removing all the jumps from asset prices.  In other words, the estimations \eqref{VaR-1} and \eqref{VaR-2}, as well as \eqref{VaR}, are two-step approaches. Algorithm of the first step of the approaches is set up immediately after this remark.  
		
		As a comparison, in the model derived without jump process \cite{FSU}, the corresponding $ \boldsymbol{\mu}^* $, $ \boldsymbol{\sigma}^* $, and $ \boldsymbol{\Sigma}^* $ are calculated from the raw data. The value-at-risk without jumps ($ \text{VaR}^* $) is 
		\begin{align} \label{VaR-3}
			\text{VaR}_\alpha^* (t)
			= 
			\exp \bigg( 
			\boldsymbol{w}^\top \Big[\boldsymbol{\mu}^* - \frac{1}{2} \boldsymbol{\sigma}^*\circ\boldsymbol{\sigma}^*\Big]t
			+ \sqrt{t}\big( \boldsymbol{w}^\top \boldsymbol{\Sigma}^* \boldsymbol{w} \big)^{\frac{1}{2}} \Phi^{-1}(1-\alpha)
			\bigg).
		\end{align}

	\end{remark}
	
	We set up the algorithm of removing jumps in the following. We define the jump days by the change-points of time series in statistics.
Using appropriate statistics method, we first detect the  jump days and denote them as $ \{\text{jpt}_1, \text{jpt}_2, \dots, \text{jpt}_m\} $.
	Then we calculate the jump size $ J_k \in \mathbb{R} $ of the $ k $-th jump day of a single instrument (asset), $ k = 1, \dots, m $, by
	\begin{align} \label{Jump}
		J_k = S(\text{jpt}_k) - S(\text{jpt}_k - 1),
	\end{align}
	where $ S(\cdot) $ is the asset closing price and $ \text{jpt}_k $ is the $ k $-th jump day.
	After that, we define the cumulative jump size $ CJ \in \mathbb{R} $ of asset closing price by
	\begin{align} \label{CJ}
		CJ_k = \sum_{l=1}^k J_l.
	\end{align}
	We finally remove all the jumps by the following:
	\begin{align} \label{Remove}
		\Big( 
		& S(0), S(1), \dots, S(\text{jpt}_1 - 1), \nonumber \\
		& \big(S(\text{jpt}_1) - CJ_1\big), \big(S(\text{jpt}_1+1) - CJ_1\big), \dots, \big(S(\text{jpt}_2 - 1) - CJ_1\big), \nonumber \\
		& \big(S(\text{jpt}_2) - CJ_2\big), \big(S(\text{jpt}_2+1) - CJ_2\big), \dots, \big(S(\text{jpt}_3 - 1) - CJ_2\big), \nonumber \\
		& \cdots \cdots \cdots \cdots \nonumber \\
		& \big(S(\text{jpt}_{m-1}) - CJ_{m-1}\big), \big(S(\text{jpt}_{m-1}+1) - CJ_{m-1}\big), \dots, \big(S(\text{jpt}_m - 1) - CJ_{m-1}\big), \nonumber \\
		& \big(S(\text{jpt}_m) - CJ_m\big), \big(S(\text{jpt}_m+1) - CJ_m\big), \dots, \big(S(T) - CJ_m\big)
		\Big)
	\end{align}
	where $ T $ is terminal time.

	\section{Application of Value-at-Risk for corn and soybean}
\label{sec3}
	\subsection{Data description}
	The daily price data for export corn and soybean prices were taken from Eikon Refinitiv of Thompson Reuters for the period from 2016-09-26 to 2023-09-14.
	Missing values are estimated by linear interpolation.
	The export market locations for corn were the US Gulf (USG), US Pacific Northwest (PNW), Ukraine, Brazil, and Argentina.
	For soybean, the locations included USG, PNW Brazil, and Argentina.
	For each of these commodities, the origins are dominant in the international market.
	The time period was one that had some heightened volatility compared to history, which makes it appropriate for our analysis. Important, this period included the following events:  COVID-19, inflation, rapid increases in oil prices, a North American drought, the emergence of Renewable diesel, Post-COVID expansion, labor shortages constraining logistical functions, other supply chain problems (congestion etc.), all of which were followed by the Russian invasion of Ukraine and the creation of and demise of the Black Sea Grain Corridor.
	Taken together, the cumulative impact of these was for a historically rapid escalation in prices, and volatility, followed by a tempered decrease in prices and volatility.
	\begin{figure}[H]
		\centering
		\begin{minipage}{\textwidth}
			\centering
			\includegraphics[width=0.7\textwidth]{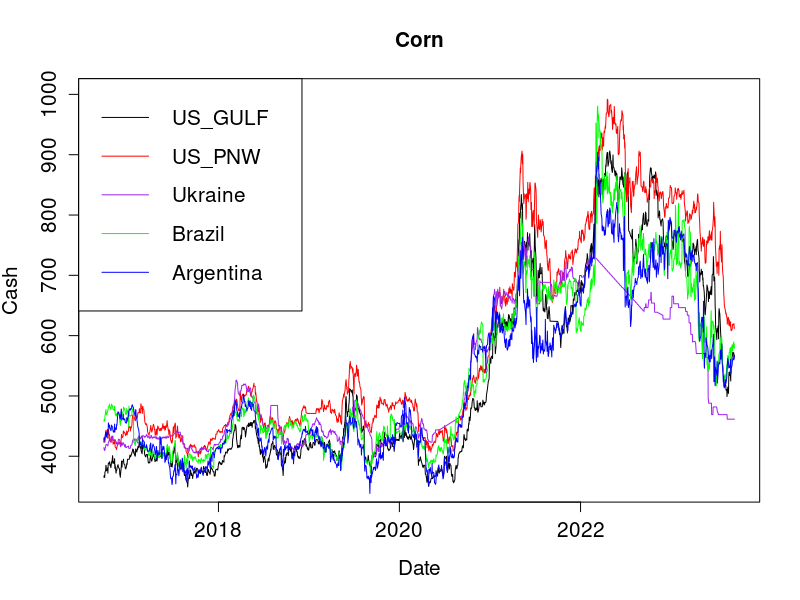}
			\caption{Daily data of cash  corn at the USG, PNW, Ukraine, Brazil, and Argentina.}
			\label{Cash_c}
		\end{minipage}
	\end{figure}
	\begin{figure}[H]
		\begin{minipage}{\textwidth}
			\centering
			\includegraphics[width=0.7\textwidth]{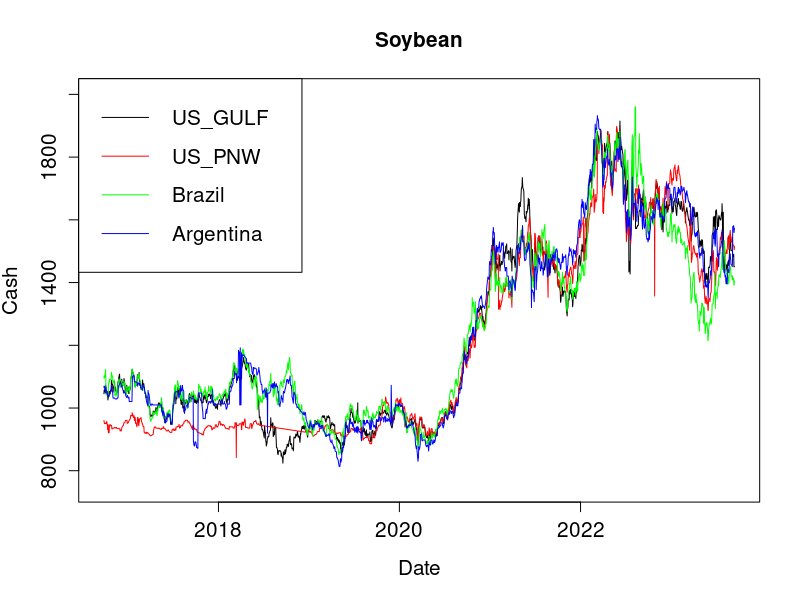}
			\caption{Daily data of cash soybean at the USG, PNW, Brazil, and Argentina.}
			\label{Cash_s}
		\end{minipage}
	\end{figure}
	\noindent 
	The multiple asset portfolio is comprised of two sets of cash assets representing the market at different locations. 
	The data is plotted in Figure \ref{Cash_c} and Figure \ref{Cash_s}.
	The boxplots of daily returns for all the instruments are plotted in Figure \ref{box_c} and Figure \ref{box_s}.
	Meanwhile, the corresponding histograms are plotted in Figure \ref{hist_c} and Figure \ref{hist_s} in Appendix.
	In the boxplots, the outliers illustrate the significant jumps of daily return.
	The abnormality of the jumps of daily return is caused by the occurrence of major events, such as Russian invasion, 
	especially for corn in Brazil and Argentina, and for soybean in PNW, Brazil, and Argentina.
	\begin{figure}[H] 
		\begin{minipage}{\textwidth}
			\centering
			\includegraphics[width=0.6\textwidth]{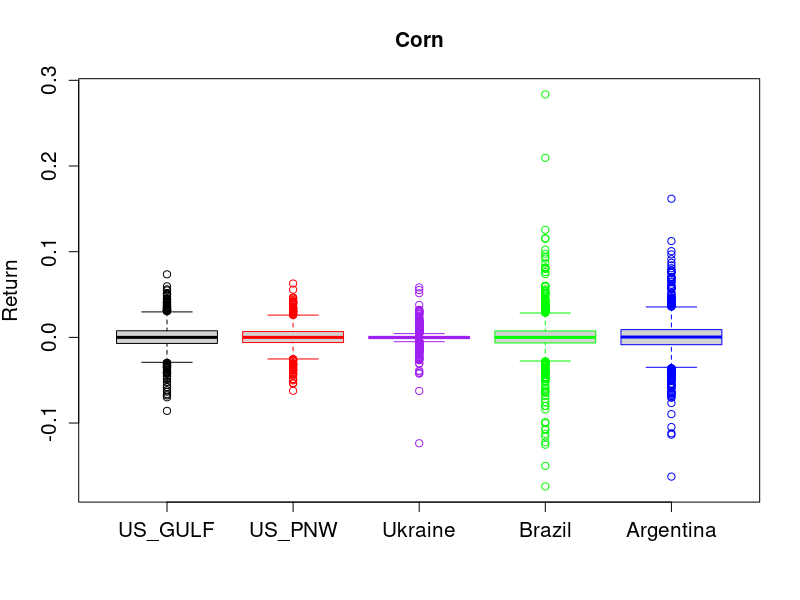}
			\caption{Boxplots of daily return of corn at the USG, PNW, Ukraine, Brazil, and Argentina.	}
			\label{box_c}
		\end{minipage}
	\end{figure}
	
	\begin{figure}[H] 
		\begin{minipage}{\textwidth}
			\centering
			\includegraphics[width=0.6\textwidth]{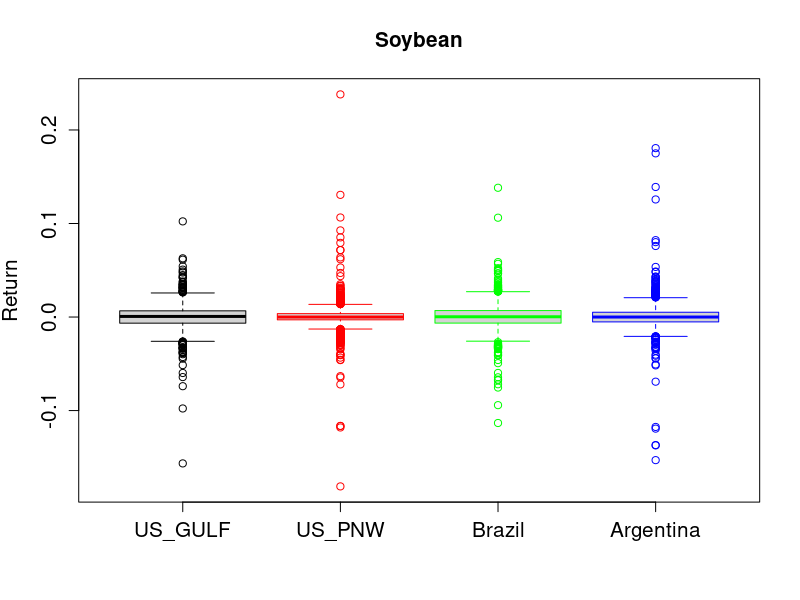}
			\caption{Boxplots of daily return of soybean at the USG, PNW, Brazil, and Argentina.}
			\label{box_s}
		\end{minipage}
	\end{figure}
	
	\subsection{A portfolio with single cash position}
	In this research, we define that jump day is the change-point identified in mean for a time series data.
	For each instrument, a jump of daily return comes from the corresponding jump of daily cash data by definition of return.
	Thus, we detect the jumps of daily cash data instead of that of daily returns.
	
	We firstly observe that this daily cash dataset is unusual, because the differences between adjacent data points are so large that almost every data point will be detected as a jump day.
	We then do a normalization to these differences.
	Next, we detect the jump days using the function "cpt.mean" \cite{Hinkley} with the method "PELT" \cite{PELT} from the package "changepoint" \cite{changepoint_1}\cite{changepoint_2} of R \cite{R}.
	Here, we note that this R package outputs the date which is one day before the jump day. 
	If we denote $ \{\text{cpt}_1, \text{cpt}_2, \dots, \text{cpt}_m\} $ as the output dates, then the jump days denoted by $ \{\text{jpt}_1, \text{jpt}_2, \dots, \text{jpt}_m\} $ are illustrated in the sequence of time $ t $ as below
	\begin{align*}
		\big( \dots, \text{cpt}_1, \text{jpt}_1, \dots, \text{cpt}_2, \text{jpt}_2, \dots \dots, \text{cpt}_m, \text{jpt}_m, \dots \big). 
	\end{align*}
	Based on the detected jump days, we let the value of the date with no-jump be $ 0 $ and let the value of the date with jump be $ 1 $, i.e. we convert the cash data into a categorical data.
	After that, we calculate the correlations of the jumps for all pairs of markets for both corn and soybean.
	Results are shown in Figure \ref{Cor}.
	\begin{figure}[H]
		\centering
		\begin{minipage}{\textwidth}
			\centering
			\includegraphics[width=0.5\textwidth]{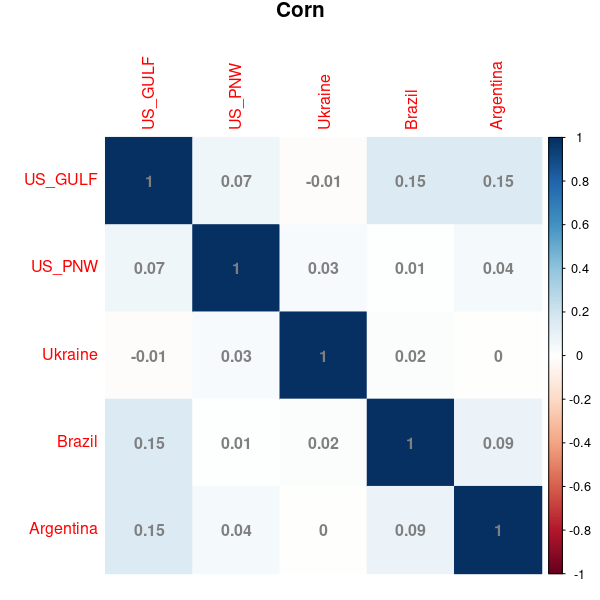}
			\includegraphics[width=0.492\textwidth]{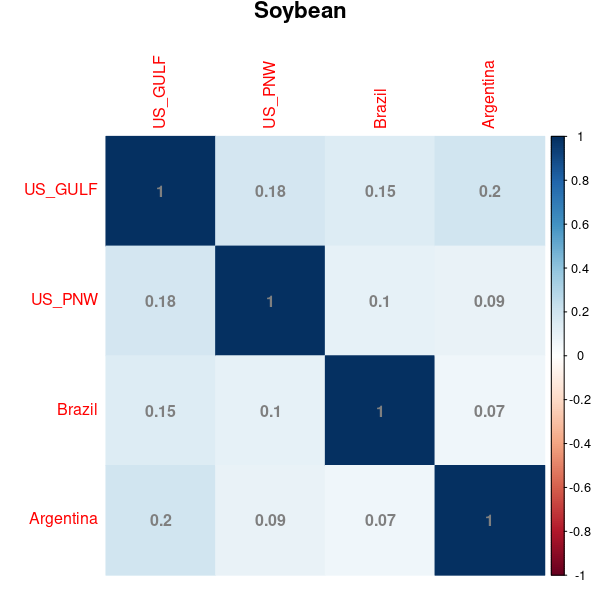}
			\caption{Correlations of the jumps of corn's and soybean's daily cash data for all pairs of markets.}
			\label{Cor}
		\end{minipage}
	\end{figure}
	In Figure \ref{Cor}, it is clear that all the correlations of jumps between different markets are close to $ 0 $. We note that, although there are two cases of correlation in corn having $ 0.15 $ and three cases of correlation in soybean having from $ 0.15 $ to $ 0.2 $, all of these values are still too low to be considered as close to $ 0 $. It implies the occurrence of the jump of closing price in one market is uncorrelated to that in a different market for both corn and soybean. Hence we only need to consider the case of a portfolio with single cash position, i.e. single asset portfolio.

	\subsection{Correction for the overestimation and underestimation of VaR}\label{sec3.3}

	In the case of single asset portfolio, we calculate $ \mathbb{E}(\text{VaR}_{0.99})(t) $ using formula \eqref{VaR-1}.
	Here we set the confidence level $ \alpha = 0.99$ to calculate the quantile $ \Phi^{-1}(1-\alpha) $.
	The time $ t $ is the total number of trading days from 2016-09-27 to the date that one wants to calculate its VaR.
	The sample mean of return $ \mu $ and the sample standard deviation of return $ \sigma $ are calculated by historical daily cash data over the time period defined in $ t $, after removing all the jumps of daily cash data.

	For comparison, we also calculate $ \text{VaR}_{0.99}^*(t) $ using formula \eqref{VaR-3} from the raw data. 
	The averages of $ \mathbb{E}(\text{VaR}_{0.99})(t) $ and that of $ \text{VaR}_{0.99}^*(t) $ with respect to $ t $, where $ t $ is from 2016-09-28 to 2023-09-14, are calculated as their expected values and summarized in Table \ref{VAR_C} and Table \ref{VAR_S}.
	\begin{table}[H]
		\caption{Expectation of daily VaR with jumps $ \mathbb{E}\big(\mathbb{E}(\text{VaR}_{0.99})(t)\big) $ and that without jumps $ \mathbb{E}\big(\text{VaR}_{0.99}^*(t)\big) $ of corn at the USG, PNW, Ukraine, Brazil, and Argentina.}
		\begin{center}
			\begin{tabular}{l c c c c c}
				\hline 
				Market & US GULF & US PNW & Ukraine & Brazil & Argentina \\
				\hline 
				$ \mathbb{E}\big(\mathbb{E}(\text{VaR}_{0.99})(t)\big) $ & $ 0.6996 $ & $ 0.6823 $ & $ 0.8136 $ & $ 0.6384$ & $ 0.4191 $  \\
				$ \mathbb{E}\big(\text{VaR}_{0.99}^*(t)\big) $ & $ 0.5896 $ & $ 0.6514 $ & $ 0.7974 $ & $ 0.4123 $ & $ 0.4027 $ \\
				\hline 
				\label{VAR_C}
			\end{tabular}
		\end{center}
	\end{table}
	\begin{table}[H]
		\caption{Expectation of daily VaR with jumps $ \mathbb{E}\big(\mathbb{E}(\text{VaR}_{0.99})(t)\big) $ and that without jumps $ \mathbb{E}\big(\text{VaR}_{0.99}^*(t)\big) $ of soybean at the USG, PNW, Brazil, and Argentina.}
		\begin{center}
			\begin{tabular}{l c c c c}
				\hline 
				Market & US GULF & US PNW & Brazil & Argentina \\
				\hline 
				$ \mathbb{E}\big(\mathbb{E}(\text{VaR}_{0.99})(t)\big) $ & $ 0.6299 $ & $ 0.7059 $ & $ 0.6051 $ & $ 0.6663 $  \\
				$ \mathbb{E}\big(\text{VaR}_{0.99}^*(t)\big) $ & $ 0.5655 $ & $0.6460 $ & $ 0.5096 $ & $ 0.4162 $ \\
				\hline
				\label{VAR_S}
			\end{tabular}
		\end{center}
	\end{table}
	In Table \ref{VAR_C} and Table \ref{VAR_S}, the magnitude of $ \mathbb{E}\big(\mathbb{E}(\text{VaR}_{0.99})(t)\big) $ and $ \mathbb{E}\big(\text{VaR}_{0.99}^*(t)\big) $ measures the risk in single asset portfolio.
	For example, the average risk of corn at {Argentina} is the lowest and that at {Ukraine} is the greatest.
	The daily VaR with and without jumps generally always increases from {2020} to end of data, in Figure \ref{VAR_DD}, Figure \ref{Arg_var}, Figure \ref{VAR_CC}, and Figure \ref{VAR_SS}. 
	It is likely due to the multitude of factors causing increased volatility, notably {COVID-19} etc. (as discussed above).  
	Notably, VaRs {generally decrease} prior to 2020, then grew substantially reflecting heightened volatility, and in fact peaked out following {the Russian invasion of Ukraine in 2022}.

	We observe that all the expectations of daily VaR with jumps are {larger} than those without jumps.
	This implies the average risk in single asset portfolio is not as {low} as previously estimated, after considering the jumps.
	In other words, our model corrects the overall {under}estimation of VaR whereas the model derived without jump process {under}estimates it in average.
	This correction is also observed in the comparisons of daily VaR with and without jumps in Figure \ref{VAR_DD}, as well as Figure \ref{VAR_CC} and Figure \ref{VAR_SS} in Appendix. 
	Usually, the red color lines representing the daily VaR with jumps are above the corresponding black color lines representing the daily VaR without jump. For instance, in Figure \ref{VAR_DD}, the difference between daily VaR with and without jumps are significant.
	In addition to it, the corresponding root-mean-square errors (RMSE) are $ 0.2440$ and $0.2711$ respectively, see Table \ref{RMSE} in Appendix.
	Comparing with the maximum range of VaR (0.8) in Figure \ref{VAR_DD}, the RMSEs indicate that there are significant differences between models with and without jumps. 
	
	\begin{figure}[H]
		\begin{minipage}{\textwidth}
			\centering
			\includegraphics[width=0.496\textwidth]{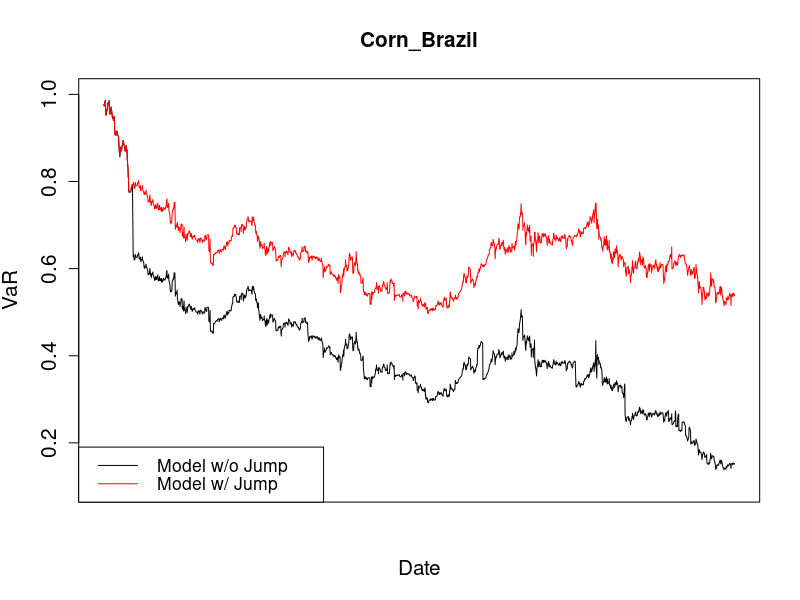}
			\includegraphics[width=0.496\textwidth]{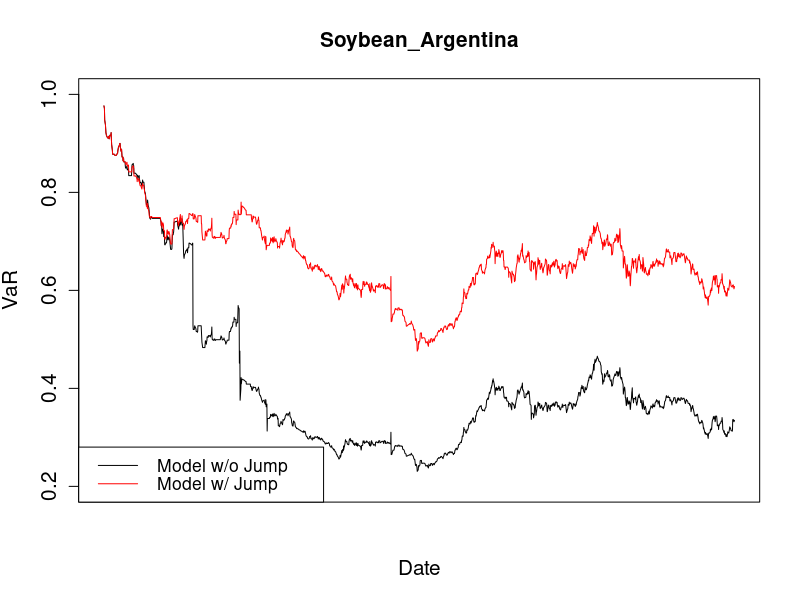}
			\caption{Daily VaR with and without jumps of corn at Brazil and those of soybean at Argentina.}
			\label{VAR_DD}
		\end{minipage}
	\end{figure}

	We perform backtesting \cite{backtest} to evaluate the performance of the VaR estimates $ \mathbb{E}(\text{VaR}_{0.99})(t) $ calculated with our model with jumps. We use the common backtesting method Dynamic Quantile (DQ) test \cite{Engle} to evaluate it for different markets.
	We then choose the function ``BacktestVaR" from the package ``GAS" (see \cite{GAS}) to perform the DQ test.
	At the significance level $ 1-\alpha = 0.01 $, the corresponding $ p $-values for the daily VaR estimates with jumps for both corn and soybean at the USG, PNW, Ukraine, Brazil, and Argentina are summarized in Table \ref{PValue}.
	All $ p $-values are larger than $ 0.01 $.
	This indicates that we fail to reject the null hypothesis which states that the current VaR violations are uncorrelated
	with past violations, at this level of significance.
	In other words, our VaR model with jumps is a correct model.

\begin{table}[H]
	\caption{DQ test $ p $-values for daily VaR estimates with jumps $ \mathbb{E}(\text{VaR}_{{0.99}})(t) $ for both corn and soybean at the USG, PNW, Ukraine, Brazil, and Argentina, at significance level 0.01.} 
	\begin{center}
		\begin{tabular}{l c c c c c}
			\hline 
			Market & US GULF & US PNW & Ukraine & Brazil & Argentina \\
			\hline 
			$ p $-value (corn) & $ 0.9334 $ & $ 0.9198 $ & $ 0.8507 $ & $ 0.8995 $ & $ 0.8941 $  \\
			$ p $-value (soybean) & $ 0.8818 $ & $ 0.8927 $ & N.A. & $ 0.8660 $ & $0.8755  $ \\
			\hline 
			\label{PValue}
		\end{tabular}
	\end{center}
\end{table}

	\begin{figure}[H]
		\begin{minipage}{\textwidth}
			\centering
			\includegraphics[width=0.6\textwidth]{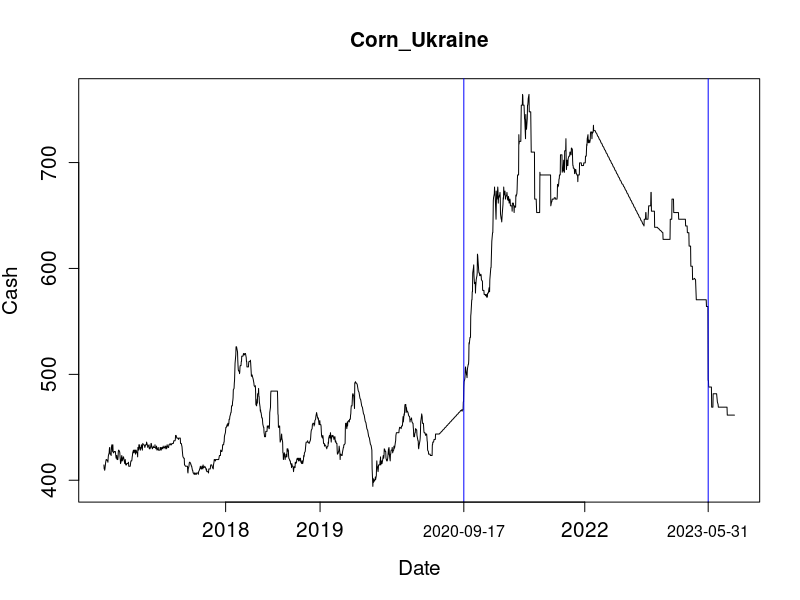}
			\caption{Daily data of cash  corn at {Ukraine}:   from 2016-09-26 to 2023-09-14.}
			\label{Arg_cash}
		\end{minipage}
	\end{figure}
	
	\begin{figure}[H]
		\begin{minipage}{\textwidth}
			\centering
			\includegraphics[width=0.6\textwidth]{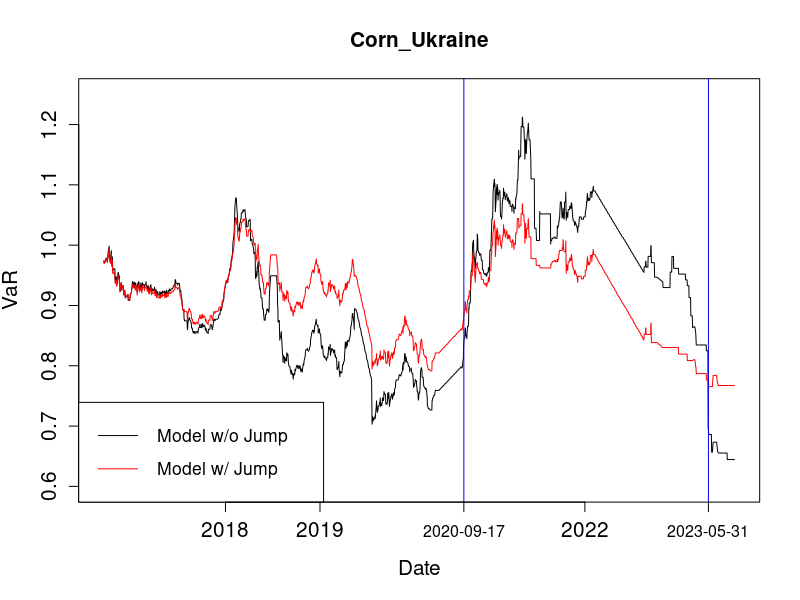}
			\caption{Daily VaR with and without jumps of corn at {Ukraine}: from 2016-09-28 to 2023-09-14.}
			\label{Arg_var}
		\end{minipage}
	\end{figure}
	Sometimes, our model estimates unusual but more reasonable VaRs which are {less} than the corresponding VaRs without jumps, after a {sharp rise} of asset closing price. For example, in Figure \ref{Arg_cash}, the closing price of corn at Ukraine rised sharply on 2020-09-17. It is well known that high-priced stocks are less risky \cite{HPLR1} \cite{HPLR2} \cite{HPLR3} \cite{HPLR4} \cite{HPLR5}.
	Hence, Figure \ref{Arg_cash} illustrates that this asset is less risky between 2020-09-17 and 2023-05-31. 	Meanwhile, in Figure \ref{Arg_var}, the VaR with jumps is generally below the corresponding VaR without jumps between 2020-09-17 and 2023-05-31.
	This implies that the VaR with jumps reveals relatively low risks for the sharp rise of asset price.  By comparison between Figure \ref{Arg_cash} and Figure \ref{Arg_var}, our model \eqref{VaR-1} corrects some overestimation of VaR without jumps.

	\subsection{Corrections by cumulative jump size}
	Our model \eqref{VaR-1} does not always estimate {smaller} daily VaRs after every {sharp rise} of asset price. Because calculation of VaR with jumps is based on the cumulative jump size defined by \eqref{CJ}, but not the jump size \eqref{Jump}.
	\begin{figure}[H]
		\begin{minipage}{\textwidth}
			\centering
			\includegraphics[width=0.6\textwidth]{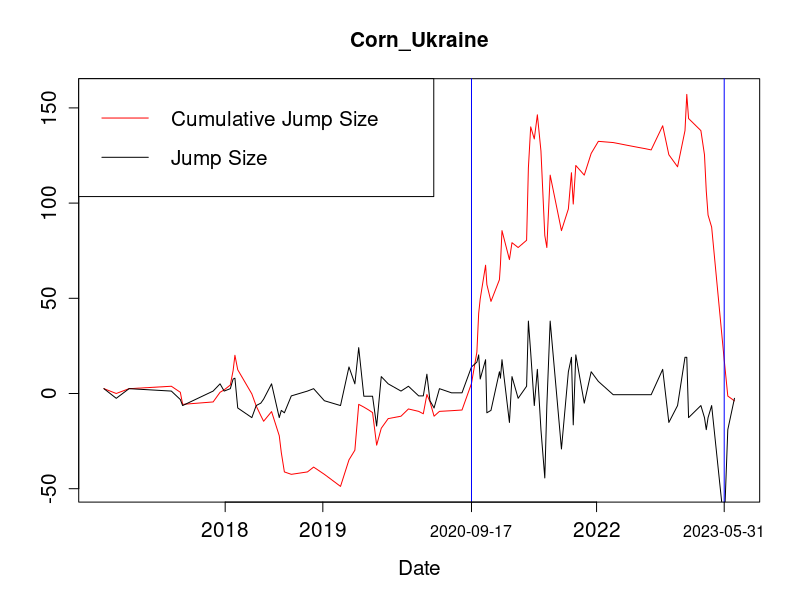}
			\caption{Jump sizes and cumulative jump sizes of cash  corn at {Ukraine}:   from 2016-09-27 to 2023-09-14.}
			\label{Arg_jump}
		\end{minipage}
	\end{figure}

\begin{table}[H]
	\caption{Part of cumulative jump sizes of cash value of corn at Ukraine.}
	\begin{center}
		\begin{tabular}{c r}
			\hline 
			Jump Day &  Cumulative Jump Size \\
			\hline
			$ \vdots $ & $ \vdots $ \\
			{2020-08-11} & $ {-8.6784} $ \\
			{2020-09-17} & $ {5.2647} $ \\
			{2020-10-08} & $ {21.7428} $ \\
			$ \vdots $ & $ \vdots $ \\
			{2023-04-12} & $ {87.3808} $ \\
			{2023-05-31} & $ {17.6655} $ \\
			{2023-06-14} & $ {-1.3477} $ \\
			$ \vdots $ & $ \vdots $ \\
			\hline
			\label{Arg_CJ}
		\end{tabular}
	\end{center}
\end{table}
	
For example, as observed in Figure \ref{Arg_jump}, the jump size of cash value of corn at Ukraine revolves around $ 0 $. 	But, its cumulative jump size increased significantly above zero on 2020-09-17 (see Table \ref{Arg_CJ}), then kept positive until 2023-05-31. By subtracting cumulative jump sizes in \eqref{Remove}, our model estimates smaller daily VaRs than the model derived without jump process, see Figure \ref{Arg_var}.
	
In contrast, the cumulative jump size decreased below zero on 2023-06-14 (see Table \ref{Arg_CJ}), then kept negative, in Figure \ref{Arg_jump}. 	By subtracting cumulative jump sizes in \eqref{Remove}, our model estimates larger daily VaRs, see Figure \ref{Arg_var}.

If the cumulative jump sizes are relatively close to zero in Figure \ref{Arg_jump}, then the VaR with and without jumps almost overlap in Figure \ref{Arg_var}. 	This implies that there is not a significant difference between the model with and without jumps.

\begin{remark}\label{remark3.1}
	In the case of single asset portfolio, 
	\begin{enumerate}
		\item if the cumulative jump size has a large positive value, the daily VaR estimated by \eqref{VaR-1} corrects the overestimation of the corresponding VaR from the model without jumps;
		\item if the cumulative jump size is negative and has a large absolute value, the daily VaR estimated by \eqref{VaR-1} corrects the underestimation of the corresponding VaR from the model without jumps;
		\item if the cumulative jump size is relatively close to $ 0 $, the above corrections are not significant.
	\end{enumerate}
\end{remark}

\begin{table}[H]
	\caption{The mean of cumulative jump sizes (MCJ) and the difference of the expectations of VaR between models with and without jumps ($ \Delta $) for corn and soybean at the USG, PNW, Ukraine, Brazil, and Argentina.}
	\begin{center}
		\begin{tabular}{l c c c c c}
			\hline 
			Market & US GULF & US PNW & Ukraine & Brazil & Argentina \\
			\hline 
			MCJ (corn) & $ -33.3239 $ & $ 8.2756 $ & $ 35.3055 $ & $ -113.2407 $ & $ 15.7472 $  \\
			$ \Delta $ (corn) & $ 0.1100 $ & $ 0.0309 $ & $ 0.0162 $ & $ 0.2261 $ & $ 0.0164 $  \\
			MCJ (soybean) & $ -44.4381 $ & $ 10.5407 $ & N.A. & $ -89.5808 $ & $ -81.0837 $ \\
			$ \Delta $ (soybean) & $ 0.0644 $ & $ 0.0599 $ & N.A. & $ 0.0955 $ & $ 0.2501 $ \\
			\hline 
			\label{VAR_D}
		\end{tabular}
	\end{center}
\end{table}

In Section \ref{sec3.3}, we have already observed that all the expectations of VaR with jumps are larger than those without jumps.
To explain this, we define the difference of the expectations of VaR between models with and without jumps by $ \Delta = \mathbb{E}\big(\mathbb{E}(\text{VaR}_{0.99})(t)\big) - \mathbb{E}\big(\text{VaR}_{0.99}^*(t)\big) $. If $ \Delta >0 $, then majority of VaR estimates from our model with jumps are larger than the corresponding VaR estimates from the model without jumps, i.e., the expectation of VaR estimated by \eqref{VaR-1} corrects the underestimation of the VaR without jumps.
After that, we calculate the mean of cumulative jump sizes (MCJ) and $ \Delta $ for both corn and soybean at the USG, PNW, Ukraine, Brazil, and Argentina.
In Table \ref{VAR_D}, we observe that $ \Delta $ for Brazil and Argentina are the largest in the market of corn and soybean, respectively. 
Meanwhile, the associated MCJs are negative and have large absolute values.   
For the remaining markets, the absolute values of MCJ are small and the associated $ \Delta $s are close to zero.
But we can observe that all the MCJ and $ \Delta $ are negatively correlated which is also supported by their correlation coefficient $ \rho(\text{MCJ}, \Delta) = -0.84 $.
Clearly, $ \Delta $ is affected by MCJ. 
Hence, similar to Remark \ref{remark3.1}, the strong negative correlation between MCJ and $ \Delta $ shows the following:

\begin{remark}
	In the case of single asset portfolio, 
	\begin{enumerate}
		\item if the mean of cumulative jump size is negative and has a large absolute value, the expectation of VaR estimated by \eqref{VaR-1} corrects the underestimation of the corresponding VaR from the model without jumps;
		\item if the absolute value of the mean of cumulative jump size is relatively small, the correction is not significant.
	\end{enumerate}
\end{remark}
\noindent The absence of the case that the mean of cumulative jump size has a large positive value will be studied in future work.
This can be caused by some financial behavior which is beyond the scope of this paper.

	\section{Conclusion}
\label{sec4}

VaR is a significant financial metric that holds relevance for all types of businesses and investment choices, regardless of their scale or magnitude. In more precise terms, the notion of VaR pertains to the computation of the utmost monetary detriment that may be incurred during a given time frame. Considering the phenomenon of the jumps of cash asset price, we define a novel VaR in this paper. The general formula of the expectation of VaR for a multi-asset portfolio with multiple cash positions is derived.  If a portfolio contains a single asset, or the returns of all the assets in a portfolio follow the same process for each jump, then the jump terms in the general formula do not explicitly exist. Under these two special cases, the corresponding formulas of the expectation of VaR can be simplified from the general formula. To estimate VaR with jumps, the sample mean and sample standard deviation of return are calculated after removing all the jumps of asset price, whereas they are calculated from the raw data of cash asset in the model derived without jump process. In the case of single asset portfolio, it is shown that the proposed model corrects both overestimation and underestimation of the VaR without jumps by the cumulative jump size.

These results have important implications for risk measurement using VaR. In this case, the results show that the VaR increased substantially during the study period as a result or numerous factors causing informational uncertainty, resulting in considerable risks for firms in the commodity marketing businesses. While all VaR measures increased during this period, the results suggest that our model with jumps estimates more reasonable VaRs. This has important implications for managing risks, margins, and other forms of controls imposed to mitigate risk exposures.

	\appendix
	
	\section{Appendix: Proof of Theorem \ref{THM}}
	\label{app1koy}

	\begin{proof}
		We let $ g(s,\zeta) $ be a simple function, i.e., $$g(s,\zeta) 	= \sum_{i=1}^n c_i \mathbbm{1}_{\mathcal{A}_i\times \mathcal{B}_i}(s,\zeta) = \sum_{i=1}^n c_i \mathbbm{1}_{\mathcal{A}_i}(s) \mathbbm{1}_{\mathcal{B}_i}(\zeta),$$
		where each $ c_i \in \mathbb{R} $, $ \mathbbm{1} $ is indicator function, the $ \mathcal{A}_i $ are disjoint Borel subsets of $ [0,t] $, and the $ \mathcal{B}_i $ are disjoint Borel subsets of $ \mathbb{R}\setminus\{0\} $.
		We let $ \mathcal{A}_i = [s_i, s_i+\Delta s_i) $, $ s_i\in[0,t] $, such that $ \cup_{i=1}^n \mathcal{A}_i = [0,t] $.
		Then
		\begin{align*}
			& \mathbb{E} \bigg( \exp \bigg( \int_0^t \int_\mathbb{R} g(s,\zeta)  N(ds,d\zeta) \bigg) \bigg) \\
			& = \mathbb{E} \bigg( \exp \bigg( \sum_{i=1}^n c_i N(\mathcal{A}_i,\mathcal{B}_i) \bigg) \bigg)
			= \mathbb{E} \bigg( \prod_{i=1}^n \exp \big(  c_i N(\mathcal{A}_i,\mathcal{B}_i) \big) \bigg) \\
			& = \mathbb{E} \bigg( \prod_{i=1}^n \exp \Big(  c_i N\big([s_i, s_i+\Delta s_i),\mathcal{B}_i\big) \Big) \bigg) \\
			& = \mathbb{E} \bigg( \prod_{i=1}^n \exp \Big( c_i \big[ N(s_i+\Delta s_i,\mathcal{B}_i) - N(s_i,\mathcal{B}_i) \big] \Big) \bigg),\ \text{by Theorem 1.5(2) in \cite{Oksendal_2}}, \\
			& = \prod_{i=1}^n \mathbb{E} \bigg( \exp \Big( c_i \big[ N(s_i+\Delta s_i,\mathcal{B}_i) - N(s_i,\mathcal{B}_i) \big] \Big) \bigg),\ \text{by Theorem 2.3.5(2) in \cite{Applebaum}}.
		\end{align*}
		By Theorem 1.5(4) in \cite{Oksendal_2}, $ N(s_i,\mathcal{B}_i) $ is a Poisson process with intensity $ \nu(\mathcal{B}_i) $. It follows that 
		\begin{align*}
			& \prod_{i=1}^n \mathbb{E} \bigg( \exp \Big( c_i \big[ N(s_i+\Delta s_i,\mathcal{B}_i) - N(s_i,\mathcal{B}_i) \big] \Big) \bigg) \\
			& = \prod_{i=1}^n \prod_{k=0}^\infty \frac{\big[\nu(\mathcal{B}_i) \Delta s_i\big]^k}{k!} \exp\big(-\nu(\mathcal{B}_i)\Delta s_i\big) \cdot \exp(c_i k) \\
			& = \prod_{i=1}^n \exp\big(-\nu(\mathcal{B}_i)\Delta s_i\big)
			\prod_{k=0}^\infty \frac{\big[\exp(c_i) \nu(\mathcal{B}_i) \Delta s_i \big]^k}{k!} \\
			& = \prod_{i=1}^n \exp\big(-\nu(\mathcal{B}_i)\Delta s_i\big)
			\exp \big(\exp(c_i) \nu(\mathcal{B}_i) \Delta s_i \big) \\
			& = \prod_{i=1}^n \exp\Big( \big[\exp(c_i)-1\big] \nu(\mathcal{B}_i)\Delta s_i \Big)
			= \exp\bigg( \sum_{i=1}^n \big[\exp(c_i)-1\big] \nu(\mathcal{B}_i)\Delta s_i \bigg) \\
			& = \exp\bigg( \int_0^t \int_\mathbb{R} \Big[\exp\big(g(s,\zeta)\big)-1\Big] \nu(d\zeta)ds \bigg).
		\end{align*}
		Hence
		\begin{align*}
			\mathbb{E} \bigg( \exp \bigg( \int_0^t \int_\mathbb{R} g(s,\zeta)  N(ds,d\zeta) \bigg) \bigg)
			= \exp\bigg( \int_0^t \int_\mathbb{R} \Big[\exp\big(g(s,\zeta)\big)-1\Big] \nu(d\zeta)ds \bigg).
		\end{align*}
		
		For every $ f \in L_1 $, there is a sequence of simple functions $ (g_n)=(g_1, g_2, g_3,\dots) $ converging to $ f $ in $ L_1 $. 
		Then there is a subsequence of $ (g_n) $ that converges to $ f $ almost everywhere. 
		It follows that this subsequence is dominated by a simple function from $ (g_n) $.
		Using Lebesgue's dominated convergence theorem, we obtain the required result by passing to the limit along this subsequence in the above.
	\end{proof}
	
	\section{Appendix: Figures}
	\label{app2koy}
	
	\begin{figure}[H]
		\centering
		\begin{minipage}{\textwidth}
			\centering
			\includegraphics[width=1.0\textwidth]{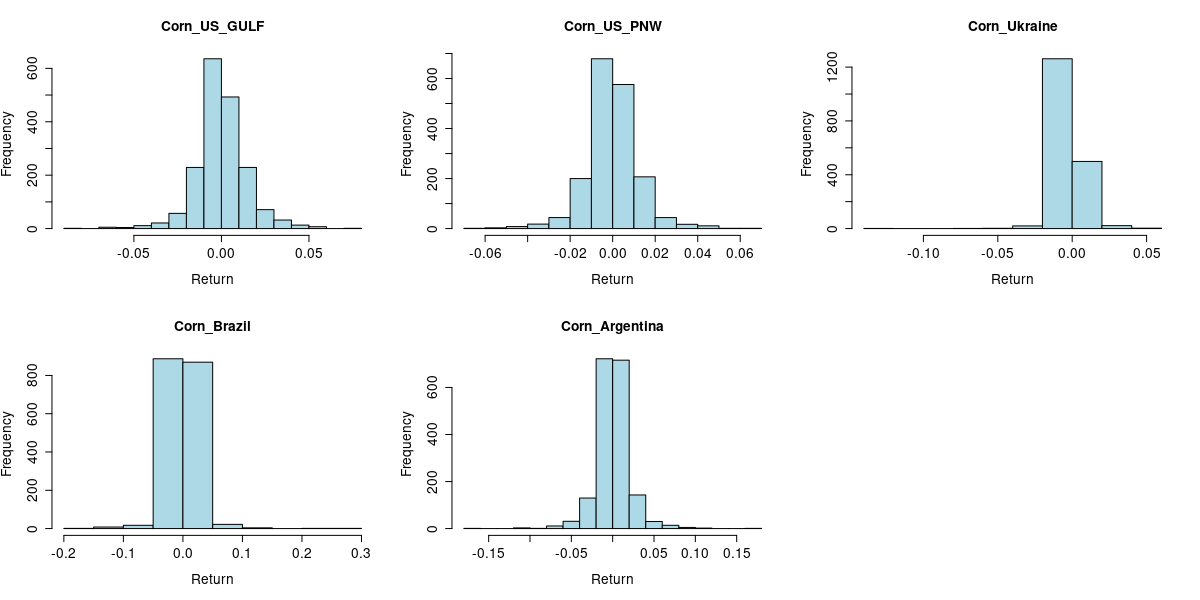}
			\caption{
			Histograms of daily return of corn at the USG, PNW, Ukraine, Brazil, and Argentina:   from 2016-09-27 to 2023-09-14.}
			\label{hist_c}
		\end{minipage}
	\end{figure}
	\begin{figure}[H] 
		\begin{minipage}{\textwidth}
			\centering
			\includegraphics[width=0.57\textwidth]{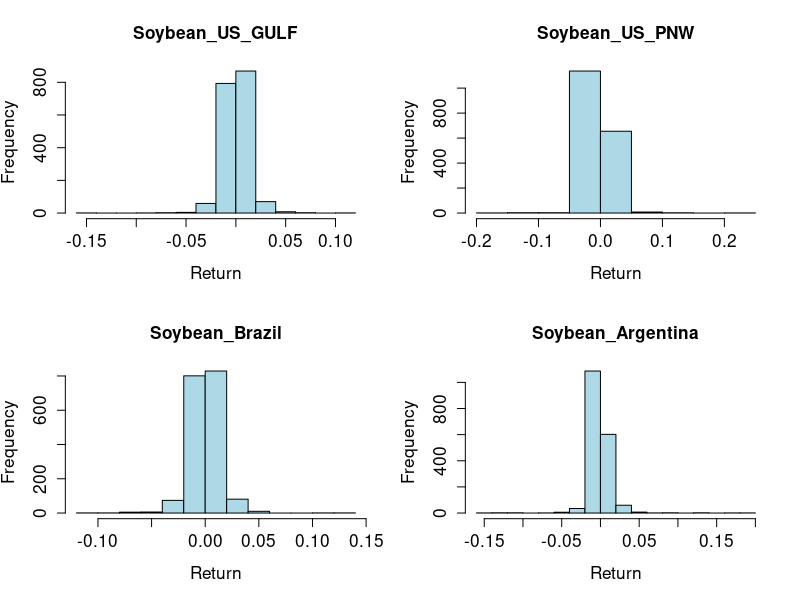}
			\caption{Histograms of daily return of soybean at the USG, PNW, Brazil, and Argentina:   from 2016-09-27 to 2023-09-14.}
			\label{hist_s}
		\end{minipage}
	\end{figure}
	
	\begin{figure}[H]
		\begin{minipage}{\textwidth}
			\centering
			\includegraphics[width=0.496\textwidth]{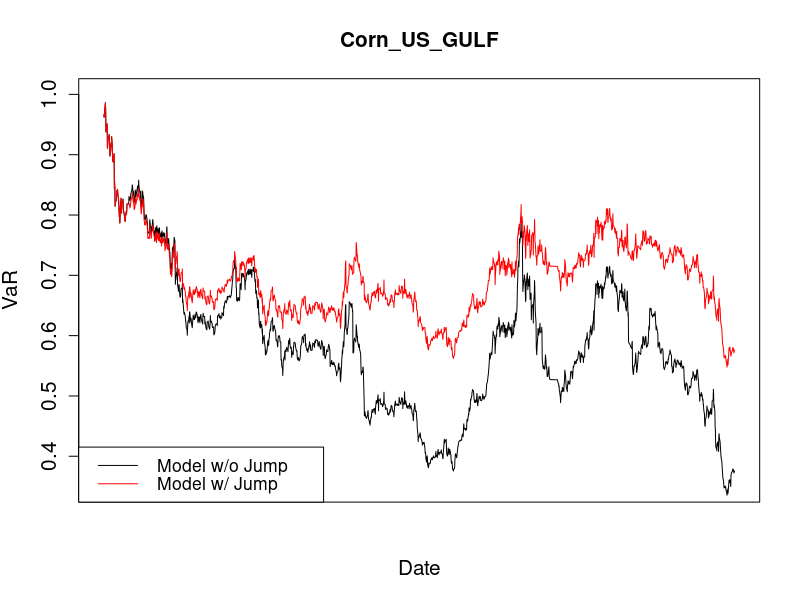}
			\includegraphics[width=0.496\textwidth]{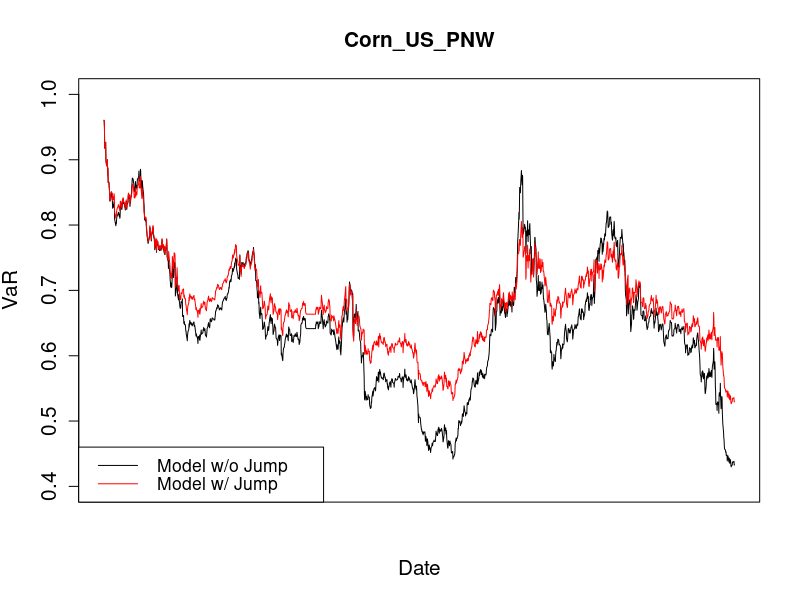} \\
			\includegraphics[width=0.496\textwidth]{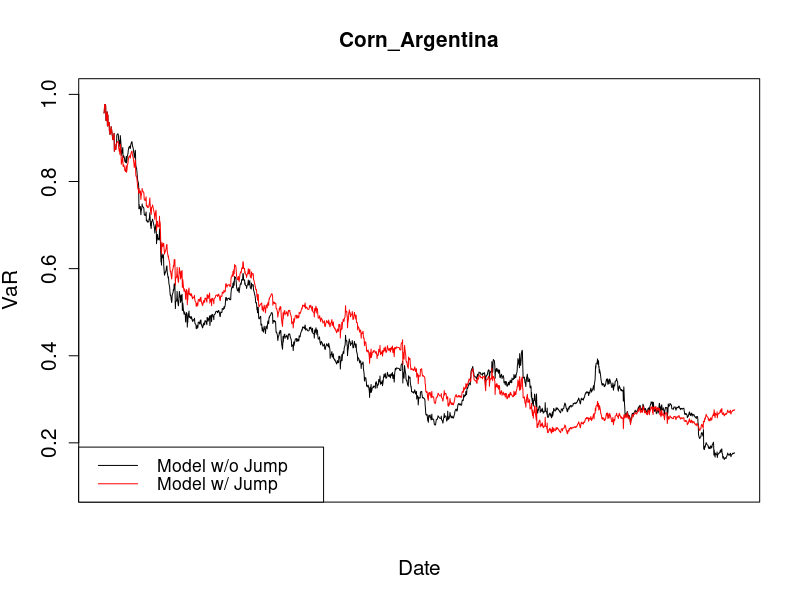}
			\caption{Daily VaR with and without jumps of corn at the USG, PNW, and {Argentina}: from 2016-09-28 to 2023-09-14.}
			\label{VAR_CC}
		\end{minipage}
	\end{figure}
	
	\begin{figure}[H]
		\begin{minipage}{\textwidth}
			\centering
			\includegraphics[width=0.496\textwidth]{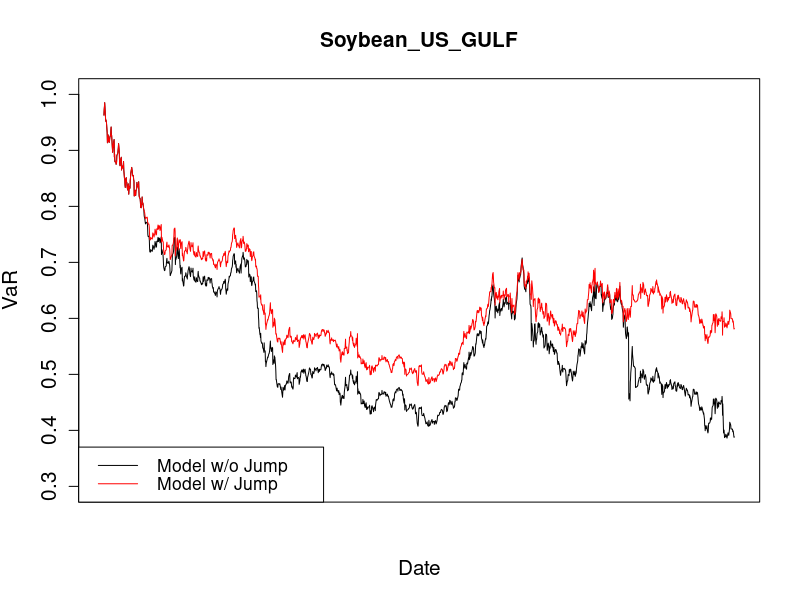}
			\includegraphics[width=0.496\textwidth]{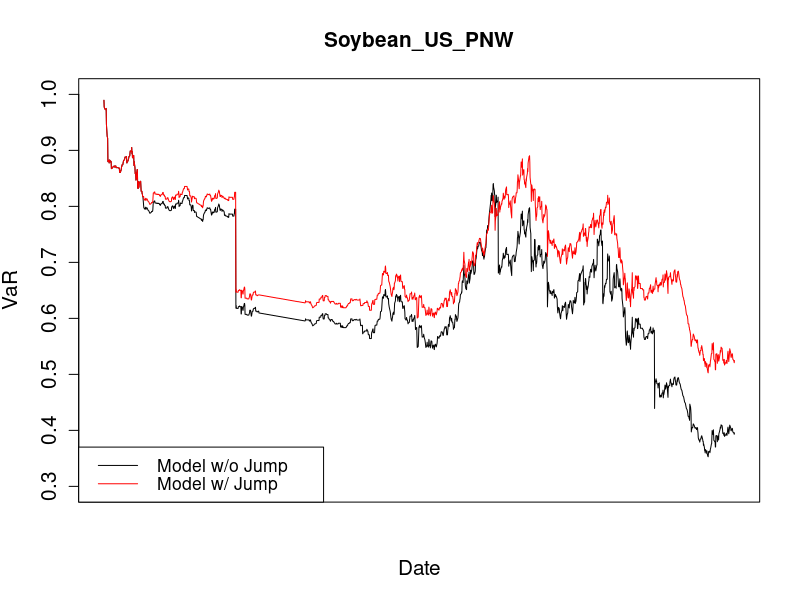} \\
			\includegraphics[width=0.496\textwidth]{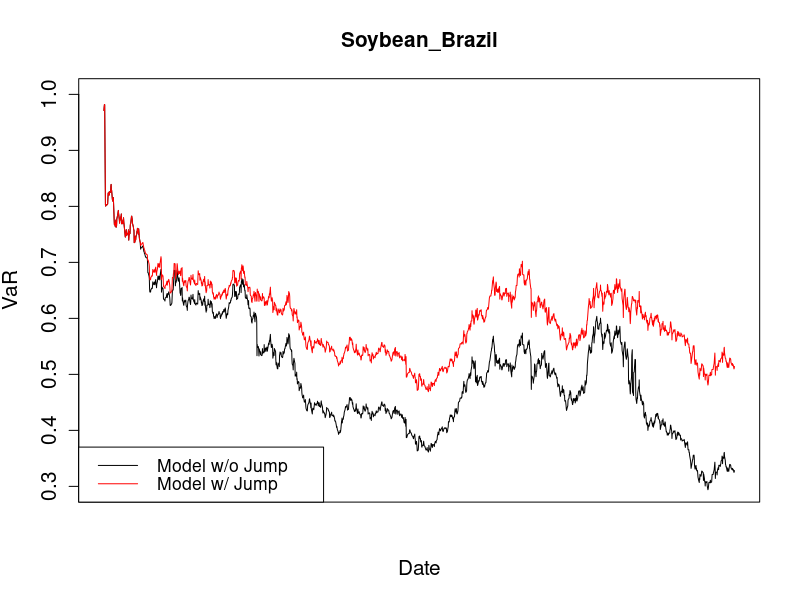}
			\caption{Daily VaR with and without jumps of soybean at the USG, PNW, and Brazil: from 2016-09-28 to 2023-09-14.}
			\label{VAR_SS}
		\end{minipage}
	\end{figure}

	\begin{table}[H]
		\caption{The root-mean-square error (RMSE) between daily VaR with and without jumps of corn and soybean at the USG, PNW, Ukraine, Brazil, and Argentina.}
		\begin{center}
			\begin{tabular}{l c c c c c}
				\hline 
				Market & US GULF & US PNW & Ukraine & Brazil & Argentina \\
				\hline 
				RMSE (corn) & $ 0.1315 $ & $ 0.0445 $ & $ 0.0635 $ & $ 0.2440 $ & $ 0.0487 $ \\
				RMSE (soybean) & $ 0.0805 $ & $ 0.0784 $ & N.A. & $ 0.1090$ & $ 0.2711$  \\
				\hline 
				\label{RMSE}
			\end{tabular}
		\end{center}
	\end{table}
	
\noindent \textbf{Acknowledgments}: We thank the anonymous referees for their constructive comments and suggestions, which helped us improve the manuscript.

\end{document}